\documentclass[onecolumn, a4size, 11pt]{IEEEtran}
\usepackage{amsmath}
\usepackage{amssymb}
\usepackage{amsfonts}
\usepackage{graphicx}
\usepackage{epsfig}
\usepackage{subfigure}
\usepackage{psfrag}

\title{Cooperative Interference Management with MISO Beamforming \footnote{This work
was presented in part at IEEE Wireless Communications and Networking
Conference (WCNC), Sydney, Australia, April 18-21, 2010.}
\footnote{R. Zhang is with the Institute for Infocomm Research,
A*STAR, Singapore and the Department of Electrical and Computer
Engineering, National University of Singapore, Singapore (e-mail:
rzhang@i2r.a-star.edu.sg).} \footnote{S. Cui is with the Department
of Electrical and Computer Engineering, Texas A\&M University,
Texas, USA (e-mail: cui@ece.tamu.edu).}}

\author{Rui Zhang and Shuguang Cui}

\begin{document}
\maketitle \thispagestyle{empty}

\vspace{-0.3in}

\begin{abstract}
This correspondence studies the downlink transmission in a
multi-cell system, where multiple base stations (BSs) each with
multiple antennas cooperatively design their respective transmit
beamforming vectors to optimize the overall system performance. For
simplicity, it is assumed that all mobile stations (MSs) are
equipped with a single antenna each, and there is one active MS in
each cell at one time. Accordingly, the system of interests can be
modeled by a multiple-input single-output (MISO) interference
channel (IC), termed as MISO-IC, with interference treated as noise.
We propose a new method to characterize different rate-tuples for
active MSs on the Pareto boundary of the achievable rate region for
the MISO-IC, by exploring the relationship between the MISO-IC and
the cognitive radio (CR) MISO channel. We show that each
Pareto-boundary rate-tuple of the MISO-IC can be achieved in a
decentralized manner when each of the BSs attains its own channel
capacity subject to a certain set of interference-power constraints
(also known as interference-temperature constraints in the CR
system) at the other MS receivers. Furthermore, we show that this
result leads to a new decentralized algorithm for implementing the
multi-cell cooperative downlink beamforming.
\end{abstract}

\begin{keywords}
Beamforming, cooperative multi-cell system, interference channel,
multi-antenna, Pareto optimal, rate region.
\end{keywords}

\setlength{\baselineskip}{1.3\baselineskip}
\newtheorem{definition}{\underline{Definition}}[section]
\newtheorem{fact}{Fact}
\newtheorem{assumption}{Assumption}
\newtheorem{theorem}{\underline{Theorem}}[section]
\newtheorem{lemma}{\underline{Lemma}}[section]
\newtheorem{corollary}{Corollary}
\newtheorem{proposition}{\underline{Proposition}}[section]
\newtheorem{example}{\underline{Example}}[section]
\newtheorem{remark}{\underline{Remark}}
\newtheorem{algorithm}{\underline{Algorithm}}[section]
\newcommand{\mv}[1]{\mbox{\boldmath{$ #1 $}}}

\section{Introduction}

Conventional wireless mobile networks are designed with a cellular
architecture, where base stations (BSs) from different cells control
communications for their associated mobile stations (MSs)
independently. The resulting inter-cell interference is treated as
additive noise and minimized by applying a predesigned frequency
reuse pattern such that the same frequency band is reused only by
non-adjacent cells. Due to the rapidly growing demand for high-rate
wireless multimedia applications, conventional cellular networks
have been pushed towards their throughput limits. Consequently, many
beyond-3G wireless technologies such as WiMAX and 3GPP UMTS Long
Term Evolution (LTE) have relaxed the constraint on the frequency
reuse such that the total frequency band becomes available for reuse
by all cells. However, this factor-one frequency reuse pattern
renders the overall network performance limited by the inter-cell
interference; consequently, more sophisticated interference
management techniques with multi-cell cooperation become crucial.
Among others, one effective method to cope with the inter-cell
interference in the cellular network is via joint signal processing
across different BSs. In this correspondence, we study a particular
type of multi-BS cooperation for the downlink transmission, where we
are interested in evaluating the benefit in terms of network
throughput by cooperatively optimizing the transmit beamforming
vectors for different BSs each with multiple antennas. Notice that
the problem setup of our work is different from that for a fully
cooperative multi-cell system considered in, e.g.,
\cite{Shamai01}-\cite{Rui}, where a central processing unit is
assumed with the global knowledge of all the required downlink
channels and user messages to jointly design the transmitted signals
for all BSs. In contrast, our work focuses on the decentralized
implementation of the multi-cell cooperative downlink beamforming
assuming only the local message and neighboring-channel knowledge at
each BS, which is more practical than implementing the full
baseband-level coordination. It is worth noting that decentralized
multi-cell cooperative downlink beamforming has been studied in
\cite{Yu} to minimize the total power consumption of all BSs to meet
with MSs' individual signal-to-interference-plus-noise ratio (SINR)
targets, based on the uplink-downlink beamforming duality. In this
work, we provide a different design approach for rate-optimal
strategies in decentralized multi-cell cooperative beamforming.

For the purpose of exposition, in this work we consider a simplified
scenario, where each MS is equipped with a single antenna, and at
any given time there is only one active MS in each cell (over a
particular frequency band). Accordingly, we can model the multi-cell
cooperative downlink transmission system as a multiple-input
single-output (MISO) Gaussian interference channel (IC), termed as
MISO-IC. From an information-theoretic viewpoint, the capacity
region of the Gaussian IC, which constitutes all the simultaneously
achievable rates for all users, is still unknown in general
\cite{Han}, while significant progresses have recently been made on
approaching this limit \cite{Tse}. Capacity-approaching techniques
for the Gaussian IC in general require certain signal-level
encoding/decoding cooperations among the users, while a more
pragmatic approach that leads to suboptimal achievable rates of the
users is to allow only single-user encoding and decoding by treating
the interference from other users as additive Gaussian noise. In
this work, we adopt the latter approach to study the design of
cooperative transmit beamforming for the MISO-IC. Particularly, we
focus on the design criterion to achieve different rate-tuples for
the users on the Pareto boundary of the achievable rate region for
the MISO-IC. Due to the coupled signal structure, the achievable
rate region for the MISO-IC with interference treated as noise is in
general a non-convex set,\footnote{It is noted that the non-convex
rate region is obtained without time-sharing (convex-hull operation)
between different achievable rate-tuples. With time-sharing, the
achievable rate region will become a convex set.} which renders the
joint optimization of beamforming vectors to achieve different
Pareto-boundary rate-tuples a challenging task. Note that this
problem has been studied in \cite{Larsson}, where for the special
two-user case, it was shown that the optimal transmit beamforming
vector to achieve a Pareto-boundary rate-pair for the MISO-IC can be
expressed as a linear combination of the zero-forcing (ZF) and
maximum-ratio transmission (MRT) beamformers. The rate maximization
for the IC with interference treated as noise has also been studied
in the literature via various ``pricing'' algorithms (see, e.g.,
\cite{Schimidt} and references therein), while in general the
price-based approach does not achieve the Pareto-optimal rates for
the MISO-IC. In \cite{Jafar}, the maximum sum-rate for the Gaussian
IC with interference treated as noise is characterized in terms of
degrees of freedom (DoF) over the interference-limited regime.

In this correspondence, we develop a new {\it parametrical}
characterization of the Pareto boundary for the MISO-IC in terms of
the interference-power levels at all receivers caused by different
transmitters, also known as the {\it interference temperature} (IT)
levels in the newly emerging ``cognitive radio (CR)'' type of
applications \cite{Haykin}. We show that each Pareto-boundary
rate-tuple can be achieved in a decentralized manner when each of
the users maximizes its own MISO channel capacity subject to a
certain set of IT constraints at the other users' receivers, which
is identical to the CR MISO channel transmit optimization problem
studied in \cite{Zhang08} and thus shares the same solution
structure. We derive new closed-form solutions for the optimal
transmit covariance matrices of all users to achieve an arbitrary
rate-tuple on the Pareto boundary of the MISO-IC rate region, from
which we see that the optimal transmit covariance matrices should
all be {\it rank-one} (i.e., beamforming is optimal).\footnote{We
thank the anonymous reviewer who brought our attention to
\cite{Shang}, in which the authors also showed the optimality of
beamforming to achieve the Pareto-boundary rates for the Gaussian
MISO-IC with interference treated as noise, via a different proof
technique.} Furthermore, we derive the conditions that are necessary
for any particular set of mutual IT constraints across all users to
guarantee a Pareto-optimal rate-tuple for the MISO-IC. Based on
these conditions, we propose a new {\it decentralized} algorithm for
implementing the multi-cell cooperative downlink beamforming. For
this algorithm, all different pairs of BSs independently search for
their mutually desirable IT constraints (with those for the MSs
associated with the other BSs fixed), under which their respective
beamforming vectors are optimized to maximize the individual
transmit rates. This algorithm improves the rates for the BSs in a
pairwise manner until the transmit rates for all BSs converge with
their mutual IT levels.

{\it Notation}: $\mv{I}$ and $\mv{0}$ denote the identity matrix and
the all-zero matrix, respectively, with appropriate dimensions. For
a square matrix $\mv{S}$, ${\rm Tr}(\mv{S})$, $|\mv{S}|$,
$\mv{S}^{-1}$, and $\mv{S}^{1/2}$ denote the trace, determinant,
inverse, and square-root of $\mv{S}$, respectively; and
$\mv{S}\succeq 0$ means that $\mv{S}$ is positive semi-definite
\cite{Boyd}. ${\rm Diag}(\mv{a})$ denotes a diagonal matrix with the
diagonal elements given by $\mv{a}$. For a matrix $\mv{M}$ of
arbitrary size, $\mv{M}^{H}$, $\mv{M}^{T}$, and ${\rm Rank}(\mv{M})$
denote the Hermitian transpose, transpose, and rank of $\mv{M}$,
respectively. $\mathbb{E}[\cdot]$ denotes the statistical
expectation. The distribution of a circularly symmetric complex
Gaussian (CSCG) random vector with the mean vector $\mv{x}$ and the
covariance matrix $\mv{\Sigma}$ is denoted by
$\mathcal{CN}(\mv{x},\mv{\Sigma})$; and $\sim$ stands for
``distributed as''. $\mathbb{C}^{m \times n}$ denotes the space of
$m\times n$ complex matrices. $\|\mv{x}\|$ denotes the Euclidean
norm of a complex vector (scalar) $\mv{x}$. The $\log(\cdot)$
function is with base $2$ by default.

\section{System Model}\label{sec:system model}

We consider the downlink transmission in a cellular network
consisting of $K$ cells, each having a multi-antenna BS to transmit
independent messages to one active single-antenna MS. It is assumed
that all BSs share the same narrow-band spectrum for downlink
transmission. Accordingly, the system under consideration can be
modeled by a $K$-user MISO-IC. It is assumed that the BS in the
$k$th cell, $k=1,\ldots,K$, is equipped with $M_k$ transmitting
antennas, $M_k\geq 1$. The discrete-time baseband signal received by
the active MS in the $k$th cell is then given by
\begin{align}\label{eq:signal model}
y_{k}=\mv{h}_{kk}^H\mv{x}_k+\sum_{j\neq k}^K\mv{h}_{jk}^H\mv{x}_j+
z_k
\end{align}
where $\mv{x}_k\in\mathbb{C}^{M_k\times 1}$ denotes the transmitted
signal from the $k$th BS; $\mv{h}_{kk}^H\in\mathbb{C}^{1\times M_k}$
denotes the direct-link channel for the $k$th MS, while
$\mv{h}_{jk}^H\in\mathbb{C}^{1\times M_j}$ denotes the cross-link
channel from the $j$th BS to the $k$th MS, $j\neq k$; and $z_k$
denotes the receiver noise. It is assumed that
${z}_k\sim\mathcal{CN}(0,\sigma_k^2), \forall k$, and all $z_k$'s
are independent.

We assume independent encoding across different BSs and thus
$\mv{x}_k$'s are independent over $k$. It is further assumed that a
Gaussian codebook is used at each BS and
$\mv{x}_k\sim\mathcal{CN}(\mv{0},\mv{S}_k), k=1,\ldots,K$, where
$\mv{S}_k=\mathbb{E}[\mv{x}_k\mv{x}_k^H]$ denotes the transmit
covariance matrix for the $k$th BS, with
$\mv{S}_k\in\mathbb{C}^{M_k\times M_k}$ and $\mv{S}_k\succeq 0$.
Notice that the CSCG distribution has been assumed for all the
transmitted signals.\footnote{It is worth noting that in
\cite{Jafar09} the authors point out that the CSCG distribution for
the transmitted signals is in general non-optimal for the Gaussian
IC with interference treated as noise, since it can be shown that
the complex Gaussian but not circularly symmetric distribution can
achieve larger rates than the symmetric distribution for some
particular channel realizations.} Furthermore, the interferences at
all the receivers caused by different transmitters are treated as
Gaussian noises. Thus, for a given set of transmit covariance
matrices of all BSs, $\mv{S}_1,\ldots,\mv{S}_K$, the achievable rate
of the $k$th MS is expressed as
\begin{eqnarray}\label{eq:user rates}
R_k(\mv{S}_1,\ldots,\mv{S}_K)=\log\left(1+\frac{\mv{h}_{kk}^H\mv{S}_k\mv{h}_{kk}}{\sum_{j\neq
k}\mv{h}_{jk}^H\mv{S}_j\mv{h}_{jk}+\sigma_k^2}\right).
\end{eqnarray}
Next, we define the achievable rate region for the MISO-IC to be the
set of rate-tuples for all MSs that can be simultaneously achievable
under a given set of transmit-power constraints for the BSs, denoted
by $P_1,\ldots,P_K$:
\begin{align}\label{eq:rate region}
\mathcal{R}\triangleq\bigcup_{\left\{\mv{S}_k\right\}: {\rm
Tr}\left(\mv{S}_k\right)\leq P_k, k=1,\ldots,K} \bigg\{
(r_1,\ldots,r_K): 0\leq r_k\leq R_k(\mv{S}_1,\ldots,\mv{S}_K),
k=1,\ldots,K \bigg\}.
\end{align}
The upper-right boundary of this region is called the {\it Pareto
boundary}, since it consists of rate-tuples at which it is
impossible to improve a particular user's rate, without
simultaneously decreasing the rate of at least one of the other
users. More precisely, the Pareto optimality of a rate-tuple is
defined as follows \cite{Larsson}.

\begin{definition}\label{definition:Pareto optimal}
A rate-tuple $(r_1,\ldots, r_K)$ is {\it Pareto optimal} if there is
no other rate-tuple $(r'_1,\ldots, r'_K)$ with $(r'_1,\ldots,
r'_K)\geq (r_1,\ldots, r_K)$ and $(r'_1,\ldots, r'_K)\neq
(r_1,\ldots, r_K)$ (the inequality is component-wise).
\end{definition}

In this work, we consider the scenario where multiple BSs in the
cellular network cooperatively design their transmit covariance
matrices in order to minimize the effect of the inter-cell
interference on the overall network throughput. In particular, we
are interested in the design criterion to achieve different
Pareto-optimal rate-tuples for the corresponding MISO-IC defined as
above.

It is worth noting that in prior works on characterizing the Pareto
boundary for the MISO-IC with interference treated as noise (see,
e.g., \cite{Larsson} and references therein), it has been assumed
(without proof) that the optimal transmit strategy for users to
achieve any rate-tuple on the Pareto boundary is {\it beamforming},
i.e., $\mv{S}_k$ is a rank-one matrix for all $k$'s. Under this
assumption, we can express $\mv{S}_k$ as
$\mv{S}_k=\mv{w}_k\mv{w}_k^H, k=1,\ldots,K$, where $\mv{w}_k
\in\mathbb{C}^{M_k\times 1}$ denotes the beamforming vector for the
$k$th user. Similarly as in the general case with ${\rm
Rank}(\mv{S}_k)\geq 1$, the achievable rates and rate region of the
MISO-IC with transmit beamforming (BF) can be defined in terms of
$\mv{w}_k$'s. However, it is not evident whether the BF case bears
the same Pareto boundary as the general case with ${\rm
Rank}(\mv{S}_k)\geq 1$ for the MISO-IC. In this work, we will show
that this is indeed the case (see Section \ref{sec:IT}).
Accordingly, we can choose to use the rate and rate-region
expressions in terms of either $\mv{S}_k$'s or $\mv{w}_k$'s to
characterize the Pareto boundary of the MISO-IC, for the rest of
this correspondence.

\begin{figure}
\centering
 \epsfxsize=0.7\linewidth
    \includegraphics[width=9cm]{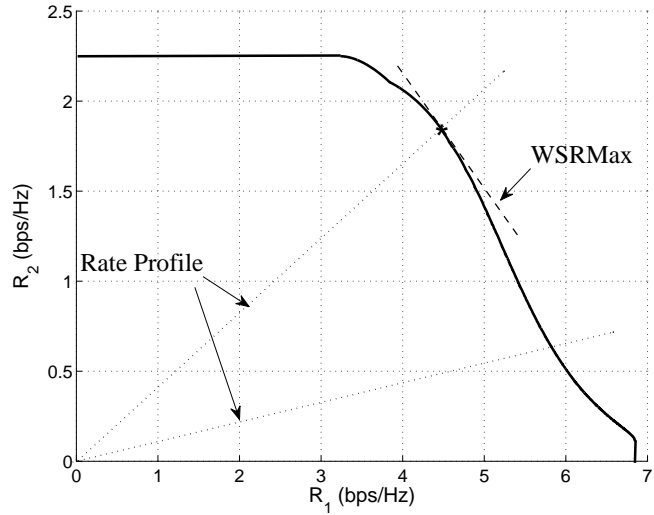}
\caption{Achievable rate region and Pareto boundary for a two-user
MISO Gaussian IC with interference treated as noise.}
\label{fig:rate region}
\end{figure}

In the following, we review some existing approaches to characterize
the Pareto boundary for the MISO-IC with interference treated as
noise. For the purpose of illustration, in Fig. \ref{fig:rate
region}, we show the achievable rate region for a two-user MISO
Gaussian IC with interference treated as noise (prior to any
time-sharing of achievable rate-pairs), which is observed to be
non-convex. A commonly adopted method to obtain the Pareto boundary
for the MISO-IC is via solving a sequence of weighted sum-rate
maximization (WSRMax) problems, each for a given set of user rate
weights, $\mu_k\geq 0, \forall k$, and given by
\begin{align}
\mathop{\mathtt{Max.}}_{\{\mv{w}_k\}} & ~~~ \sum_{k=1}^K\mu_k
\log\left(1+\gamma_k(\mv{w}_1,\ldots,\mv{w}_K)\right)\nonumber \\
\mathtt{s.t.} & ~~~ \|\mv{w}_k\|^2\leq P_k, ~k=1,\ldots,K
\end{align}
where $\gamma_k(\mv{w}_1,\ldots,\mv{w}_K)$ is the receiver SINR for
the $k$th user defined as
\begin{align}
\gamma_k(\mv{w}_1,\ldots,\mv{w}_K)=\frac{\|\mv{h}_{kk}^H\mv{w}_k\|^2}{\sum_{j\neq
k}\|\mv{h}_{jk}^H\mv{w}_j\|^2+\sigma_k^2}, ~k=1,\ldots,K.
\end{align}
It is easy to verify that this problem is non-convex, and thus
cannot be solved efficiently. In addition, the above method cannot
guarantee the finding of all Pareto-boundary points for the MISO-IC
(cf. Fig. \ref{fig:rate region}).

An alternative method to characterize the {\it complete} Pareto
boundary for the MISO-IC is based on the concept of {\it rate
profile} \cite{Zhang06}. Specifically, any rate-tuple on the Pareto
boundary of the rate region can be obtained via solving the
following optimization problem with a particular rate-profile
vector, $\mv{\alpha}=(\alpha_1,\ldots,\alpha_K)$:
\begin{align}\label{eq:SRMax}
\mathop{\mathtt{Max.}}_{R_{\rm sum}, \{\mv{w}_k\}}& ~~ R_{\rm sum} \nonumber \\
\mathtt{s.t.} & ~~
\log\left(1+\gamma_k(\mv{w}_1,\ldots,\mv{w}_K)\right)\geq
\alpha_{k}R_{\rm sum}, ~~k=1,\ldots,K \nonumber \\ & ~~
\|\mv{w}_k\|^2\leq P_k, ~~k=1,\ldots,K
\end{align}
with $\alpha_k$ denoting the target ratio between the $k$th user's
achievable rate and the users' sum-rate, $R_{\rm sum}$. Without loss
of generality, we assume that $\alpha_k\geq 0, \forall k$, and
$\sum_{k=1}^{K}\alpha_k=1$. For a given $\mv{\alpha}$, let the
optimal solution of Problem (\ref{eq:SRMax}) be denoted by $R_{\rm
sum}^*$. Then, it follows that $R_{\rm sum}^*\cdot\mv{\alpha}$ must
be the corresponding Pareto-optimal rate-tuple, which can be
geometrically viewed as (cf. Fig. \ref{fig:rate region}) the
intersection between a ray in the direction of $\mv{\alpha}$ and the
Pareto boundary of the rate region. Thereby, with different
$\mv{\alpha}$'s, solving Problem (\ref{eq:SRMax}) yields the
complete Pareto boundary for the rate region, which does not need to
be convex.

Next, we show that Problem (\ref{eq:SRMax}) can be solved via
solving a sequence of feasibility problems each for a fixed $r_{\rm
sum}$ and expressed as
\begin{align}\label{eq:feasibility problem}
\mathop{\mathtt{Find}}& ~~ \{\mv{w}_k\}\nonumber \\
\mathtt{s.t.} & ~~
\log\left(1+\gamma_k(\mv{w}_1,\ldots,\mv{w}_K)\right)\geq
\alpha_{k}r_{\rm sum}, ~~k=1,\ldots,K \nonumber \\ & ~~
\|\mv{w}_k\|^2\leq P_k, ~~k=1,\ldots,K.
\end{align}
If the above problem is feasible for a given sum-rate target,
$r_{\rm sum}$, it follows that $R_{\rm sum}^*\geq r_{\rm sum}$;
otherwise, $R_{\rm sum}^*< r_{\rm sum}$. Thus, by solving Problem
(\ref{eq:feasibility problem}) with different $r_{\rm sum}$'s and
applying the simple bisection method \cite{Boyd} over $r_{\rm sum}$,
$R_{\rm sum}^*$ can be obtained for Problem (\ref{eq:SRMax}). Let
$\bar{\gamma}_k=2^{\alpha_{k}r_{\rm sum}}-1, k=1,\ldots K$. Then,
for Problem (\ref{eq:feasibility problem}), we can replace the rate
constraints by the equivalent SINR constraints given by
\begin{align}
\gamma_k(\mv{w}_1,\ldots,\mv{w}_K)\geq \bar{\gamma}_k,
~~k=1,\ldots,K
\end{align}
Similarly as shown in \cite{Bengtsson99}, the resultant feasibility
problem can be transformed into a second-order cone programming
(SOCP) problem, which is convex and can be solved efficiently
\cite{CVX}.

\section{Characterizing Pareto Boundary for MISO-IC via Interference Temperature Control} \label{sec:IT}

In this section, instead of investigating centralized approaches, we
present a new method to characterize the Pareto boundary for the
MISO-IC in a distributed fashion, by exploring its relationship with
the CR MISO channel \cite{Zhang08}. We start with the general-rank
transmit covariance matrices $\mv{S}_k$'s for the MISO-IC. First, we
introduce a set of auxiliary variables, $\Gamma_{kj}, k=1,\ldots,K,
j=1,\ldots,K, j\neq k$, where $\Gamma_{kj}$ is called the
interference-power or interference-temperature (IT) constraint from
the $k$th BS to $j$th MS, $j\neq k$, $\Gamma_{kj}\geq 0$. For
notational convenience, let $\mv{\Gamma}$ be the vector consisting
of all $K(K-1)$ different $\Gamma_{kj}$'s, and $\mv{\Gamma}_k$ be
the vector consisting of all $2(K-1)$ different $\Gamma_{kj}$'s and
$\Gamma_{jk}$'s, $j=1,\ldots,K, j\neq k$, for any given
$k\in\{1,\ldots,K\}$.

Next, we consider a set of parallel transmit covariance optimization
problems, each for one of the $K$ BSs in the MISO-IC expressed as
\begin{align}\label{eq:CR capacity}
\mathop{\mathtt{Max.}}_{\mv{S}_k} & ~~~
\log\left(1+\frac{\mv{h}_{kk}^H\mv{S}_k\mv{h}_{kk}}{\sum_{j\neq
k}\Gamma_{jk}+\sigma_k^2}\right) \nonumber \\
\mathtt{s.t.} & ~~~ \mv{h}_{kj}^H\mv{S}_k\mv{h}_{kj}\leq
\Gamma_{kj}, ~\forall j\neq k \nonumber \\ & ~~~ {\rm
Tr}(\mv{S}_k)\leq P_k, ~\mv{S}_k\succeq 0
\end{align}
where $k\in\{1,\ldots,K\}$. Note that in the above problem for a
given $k$, $\mv{\Gamma}_k$ is fixed. For notational convenience, we
denote the optimal value of this problem as $C_k(\mv{\Gamma}_k)$. If
in the objective function of (\ref{eq:CR capacity}) we set
$\Gamma_{jk}=\mv{h}_{jk}^H\mv{S}_j\mv{h}_{jk}, \forall j\neq k$,
$C_k(\mv{\Gamma}_k)$ becomes equal to the maximum achievable rate of
an equivalent MISO CR channel \cite{Zhang08}, where the $k$th BS is
the so-called ``secondary'' user transmitter, and all the other
$K-1$ BSs, indexed by $j=1,\ldots,K, j\neq k$, are the ``primary''
user transmitters, each of which has a transmit covariance matrix,
$\mv{S}_j$, and its intended ``primary'' user receiver is protected
by the secondary user via the IT constraint:
$\mv{h}_{kj}^H\mv{S}_k\mv{h}_{kj}\leq \Gamma_{kj}$. In
\cite{Zhang08}, it was proved that the solution for Problem
(\ref{eq:CR capacity}) is rank-one, i.e., beamforming is optimal,
and in the special case of $K=2$ (i.e., one single primary user), a
closed-form solution for the optimal beamforming vector was derived.
In the following proposition, we provide a new closed-form solution
for Problem (\ref{eq:CR capacity}) with arbitrary values of $K$,
from which we can easily infer that beamforming is indeed optimal.
\begin{proposition}\label{proposition:optimal CR solution}
The optimal solution of Problem (\ref{eq:CR capacity}) is {\it
rank-one}, i.e., $\mv{S}_k=\mv{w}_k\mv{w}_k^H$, and
\begin{equation}\label{eq:optimal BF}
\mv{w}_k=\left(\sum_{j\neq
k}\lambda_{kj}\mv{h}_{kj}\mv{h}_{kj}^H+\lambda_{kk}\mv{I}\right)^{-1}\mv{h}_{kk}\sqrt{p_k}
\end{equation}
where $\lambda_{kj}$, $j\neq k$, and $\lambda_{kk}$, are
non-negative constants (dual variables) corresponding to the $k$th
BS's IT constraint for the $j$th MS and its own transmit-power
constraint, respectively, which are obtained as the optimal
solutions for the dual variables in the dual problem of Problem
(\ref{eq:CR capacity}); and $p_k$ is given by
\begin{align}
p_k=\left(\frac{1}{\ln2}-\frac{\sum_{j\neq
k}\Gamma_{jk}+\sigma_k^2}{\|\mv{A}_k\mv{h}_{kk}\|^2}\right)^+
\frac{1}{\|\mv{A}_k\mv{h}_{kk}\|^2}
\end{align}
where $\mv{A}_k\triangleq\left(\sum_{j\neq
k}\lambda_{kj}\mv{h}_{kj}\mv{h}_{kj}^H+\lambda_{kk}\mv{I}\right)^{-1/2}$
and $(x)^+\triangleq \max(0,x)$.
\end{proposition}
\begin{proof}
Please see Appendix \ref{appendix:proof optimality rank one}.
\end{proof}

Now, we are ready to present a {\it parametrical} characterization
of the Pareto boundary for the MISO-IC in terms of $\mv{\Gamma}$ as
follows.
\begin{proposition}\label{proposition:relationship}
For any rate-tuple $(R_1,\ldots,R_K)$ on the Pareto boundary of the
MISO-IC rate region defined in (\ref{eq:rate region}), which is
achievable with a set of transmit covariance matrices,
$\mv{S}_1,\ldots,\mv{S}_K$, there is a corresponding
interference-power constraint vector, $\mv{\Gamma}\geq 0$, with
$\Gamma_{kj}=\mv{h}_{kj}^H\mv{S}_k\mv{h}_{kj}, \forall j\neq k,
j\in\{1,\ldots,K\}$, and $k\in\{1,\ldots,K\}$, such that
$R_k=C_k(\mv{\Gamma}_k), \forall k$, and $\mv{S}_k$ is the optimal
solution of Problem (\ref{eq:CR capacity}) for the given $k$.
\end{proposition}
\begin{proof}
Please see Appendix \ref{appendix:proof relationship}.
\end{proof}

From Proposition \ref{proposition:relationship}, it follows that the
Pareto boundary for the MISO-IC is parameterized in terms of a
lower-dimensional manifold over the non-negative real vector
$\mv{\Gamma}$, in comparison with that over the complex transmit
covariance matrices, $\mv{S}_k$'s, or with that over the complex
beamforming vectors, $\mv{w}_k$'s. Furthermore, by combining
Propositions \ref{proposition:optimal CR solution} and
\ref{proposition:relationship}, it follows that {\it beamforming} is
indeed optimal to achieve any rate-tuple on the Pareto boundary for
the MISO-IC.

\begin{remark}
It is worth noting that the dimensionality reduction approach
proposed in this work for characterizing the Pareto boundary of the
MISO-IC is in spirit similar to that proposed in \cite{Larsson},
where it has been shown that the transmit beamforming vectors to
achieve any Pareto-boundary rate-tuple of the $K$-user MISO-IC with
interference treated as noise can be expressed in the following
forms:
\begin{align}\label{eq:old BF solution}
\mv{w}_k=\sum_{j=1}^K\xi_{kj}\mv{h}_{kj}, ~~k=1,\ldots,K
\end{align}
where $\xi_{kj}$'s are complex coefficients. Note that under the
assumption of independent $\mv{h}_{kj}$'s, the above beamforming
structure is non-trivial only when $M_k>K$. For this case, from
Remark \ref{remark:KKT} in Appendix \ref{appendix:proof optimality
rank one}, it is known that for the optimal beamforming structure
given in (\ref{eq:optimal BF}), we have $\lambda_{kk}>0$. With this
and by applying the matrix inversion lemma \cite{Horn}, it can be
shown (the detailed proof is omitted here for brevity) that the
optimal beamforming structure given by (\ref{eq:optimal BF}) is
indeed in accordance with that given by (\ref{eq:old BF solution}).
The main difference for these two methods to characterize the Pareto
boundary for the MISO-IC lies in their adopted parameters: The
method in our work uses $K(K-1)$ real $\Gamma_{kj}$'s, while that in
\cite{Larsson} uses $K(K-1)$ complex $\xi_{kj}$'s. Note that
$\Gamma_{kj}$ corresponds to the IT constraint from the $k$th user
transmitter to the $j$th user receiver, whereas there is no
practical meaning associated with $\xi_{kj}$. Consequently, as will
be shown next, the proposed method in our work leads to a practical
algorithm to implement the multi-cell cooperative downlink
beamforming, via iteratively searching for mutually desirable IT
constraints between different pairs of BSs.
\end{remark}

\section{Decentralized Algorithm for Multi-Cell Cooperative Beamforming} \label{sec:algorithm}

In this section, we develop a new {\it decentralized} algorithm that
practically implements the multi-cell cooperative downlink
beamforming based on the results derived in the previous section. It
is assumed that each BS in the cellular network has the perfect
knowledge of the channels from it to all MSs. Furthermore, it is
assumed that all BSs operate according to the same protocol
described as follows. Initially, a set of prescribed IT constraints
in $\mv{\Gamma}$ are set over the whole network, and the $k$th BS is
informed of its corresponding $\mv{\Gamma}_k, k=1,\ldots,K$.
Accordingly, each BS sets its own transmit beamforming vector via
solving Problem (\ref{eq:CR capacity}) and sets its transmit rate
equal to the optimal objective value of Problem (\ref{eq:CR
capacity}), which is achievable for its MS since the actual IT
levels from the other BSs must be below their prescribed
constraints. Then, by assuming that there is an error-free link
between each pair of BSs, all different pairs of BSs start to
communicate with each other for updating their mutual IT constraints
(the details are given later in this section), under which each pair
of BSs reset their respective beamforming vectors via solving
Problem (\ref{eq:CR capacity}) such that the achievable rates for
their MSs both get improved. Notice that each pair of updating BSs
keeps the IT constraints for the MSs associated with the other BSs
excluding this pair fixed; and as a result, the transmit rates for
all the other MSs are not affected. Therefore, the above algorithm
can be implemented in a pairwise decentralized manner across the
BSs, while it converges when there are no incentives for all
different pairs of BSs to further update their mutual IT
constraints.

Next, we focus on the key issue on how to update the mutual IT
constraints for each particular pair of BSs to guarantee the rate
increase for both of their MSs. To resolve this problem, in the
following proposition, we first provide the necessary conditions for
any given $\mv{\Gamma}\geq 0$ (component-wise) to correspond to a
Pareto-optimal rate-tuple for the MISO-IC, which will lead to a
simple rule for updating the mutual IT constraints between different
pairs of BSs. Note that from Proposition
\ref{proposition:relationship}, it follows that for any
Pareto-optimal rate-tuple of the MISO-IC, there must exist a
$\mv{\Gamma}$ such that the optimal solutions of the problems given
in (\ref{eq:CR capacity}) for all $k$'s are the same as those for
the general formulation of MISO-IC to achieve this rate-tuple.
However, for any given $\mv{\Gamma}\geq 0$, it remains unknown
whether this value of $\mv{\Gamma}$ will correspond to a
Pareto-optimal rate-tuple for the MISO-IC.

\begin{proposition}\label{proposition:conditions}
For an arbitrarily chosen $\mv{\Gamma}\geq 0$, if the optimal rate
values of the problems in (\ref{eq:CR capacity}) for all $k$'s,
$C_k(\mv{\Gamma}_k)$'s, are Pareto-optimal on the boundary of the
MISO-IC rate region defined in (\ref{eq:rate region}), then for any
pair of $(i,j), i\in\{1,\ldots,K\}, j\in\{1,\ldots,K\}$, and $i\neq
j$, it must hold that $|\mv{D}_{ij}|=0$, where $\mv{D}_{ij}$ is
defined as
\begin{eqnarray}\label{eq:D}
\mv{D}_{ij}=\left[\begin{array}{cc} \frac{\partial
C_i\left(\mv{\Gamma}_i\right)}{\partial\Gamma_{ij}} & \frac{\partial
C_i\left(\mv{\Gamma}_i\right)}{\partial\Gamma_{ji}} \vspace{0.1in} \\
\frac{\partial C_j\left(\mv{\Gamma}_j\right)}{\partial\Gamma_{ij}} &
\frac{\partial C_j\left(\mv{\Gamma}_j\right)}{\partial\Gamma_{ji}}
\end{array}\right].
\end{eqnarray}
\end{proposition}
\begin{proof}
Please see Appendix \ref{appendix:proof conditions}.
\end{proof}

Note that $\mv{D}_{ij}$'s for all different pairs of $(i,j)$ can be
obtained from the (primal and dual) solutions of the problems given
in (\ref{eq:CR capacity}) for all $k$'s with the given $\mv{\Gamma}$
(for the details, please refer to Appendix \ref{appendix:proof
optimality rank one}). More specifically, we have
\begin{align}\label{eq:a}
\frac{\partial
C_i\left(\mv{\Gamma}_i\right)}{\partial\Gamma_{ij}}=\lambda_{ij}
\end{align}
where $\lambda_{ij}$ is the solution for the dual problem of Problem
(\ref{eq:CR capacity}) with $k=i$, which corresponds to the $j$th IT
constraint, and from the objective function of Problem (\ref{eq:CR
capacity}),
\begin{align}\label{eq:b}
\frac{\partial
C_i\left(\mv{\Gamma}_i\right)}{\partial\Gamma_{ji}}=\frac{-\mv{h}_{ii}^H\mv{S}_i^{\star}\mv{h}_{ii}}{\ln2(\sum_{l\neq
i}\Gamma_{li}+\sigma_i^2) (\sum_{l\neq
i}\Gamma_{li}+\sigma_i^2+\mv{h}_{ii}^H\mv{S}_i^{\star}\mv{h}_{ii})}
\end{align}
where $\mv{S}_i^{\star}$ is the optimal solution of Problem
(\ref{eq:CR capacity}) with $k=i$. Similarly, $\frac{\partial
C_j\left(\mv{\Gamma}_j\right)}{\partial\Gamma_{ij}}$ and
$\frac{\partial C_j\left(\mv{\Gamma}_j\right)}{\partial\Gamma_{ji}}$
can be obtained from solving Problem (\ref{eq:CR capacity}) via the
Lagrange duality method with $k=j$.

From Proposition \ref{proposition:conditions}, the following
observations can be easily obtained (the proofs are omitted for
brevity):
\begin{itemize}
\item For any particular $\mv{\Gamma}$ that corresponds to a Pareto-optimal rate-tuple,
it must hold that $\Gamma_{ij}\leq \bar{\Gamma}_{ij}$, $\forall i,j,
i\neq j$, where
$\bar{\Gamma}_{ij}=\frac{\|\mv{h}_{ij}^H\mv{h}_{ii}\|^2P_i}{\|\mv{h}_{ii}\|^2}$
corresponds to the case of using maximum transmit power with MRT
beamforming for the $i$th BS;

\item For any particular $\mv{\Gamma}$ that corresponds to a Pareto-optimal rate-tuple,
it must hold that
$\mv{h}_{ij}^H\mv{S}_i^{\star}\mv{h}_{ij}=\Gamma_{ij}$, $\forall
i,j, i\neq j$, i.e., the IT constraints across all BSs must be
tight.
\end{itemize}

From the above observations, we see that if we are only interested
in the values of $\mv{\Gamma}$ that correspond to Pareto-optimal
rate-tuples for the MISO-IC, it is sufficient for us to focus on the
subset of $\mv{\Gamma}$ within the set $\mv{\Gamma}\geq 0$, in which
$\Gamma_{ij}\leq \bar{\Gamma}_{ij}$ and
$\Gamma_{ij}=\mv{h}_{ij}^H\mv{S}_i^{\star}\mv{h}_{ij}, \forall i,j,
i\neq j$.

Based on Proposition \ref{proposition:conditions}, we can develop a
simple rule for different pairs of BSs in the cooperative multi-cell
system to update their mutual IT constraints for improving both of
their transmit rates, while keeping those of the other BSs
unchanged. From the proof of Proposition
\ref{proposition:conditions} given in Appendix \ref{appendix:proof
conditions}, it follows that the method for any updating BS pair
$(i,j)$ to fulfill the above requirements is via changing
$\Gamma_{ij}$ and $\Gamma_{ji}$ according to (\ref{eq:gamma
update}). Note that in general, the choice for $\mv{d}_{ij}$ in
(\ref{eq:gamma update}) to make $\mv{D}_{ij}\mv{d}_{ij}>0$ is not
unique. For notational
conciseness, let $\mv{D}_{ij}={\small\left[\begin{array}{ll} a & b \\
c & d\end{array}\right]}$; it can then be shown that one particular
choice for $\mv{d}_{ij}$ is
\begin{align}\label{eq:d update}
\mv{d}_{ij}={\rm sign}(ad-bc)\cdot[\alpha_{ij}d-b, a-\alpha_{ij}
c]^T
\end{align}
where ${\rm sign}(x)=1$ if $x\geq 0$ and $=-1$ otherwise;
$\alpha_{ij}\geq 0$ is a constant that controls the ratio between
the rate increments for the $i$th and $j$th BSs. It can be easily
verified that when $\alpha_{ij}>1$, a larger rate increment is
resulted for the $i$th BS than that for the $j$th BS, and vice versa
when $\alpha_{ij}<1$ (provided that the step-size $\delta_{ij}$ in
(\ref{eq:gamma update}) is sufficiently small).

More specifically, the procedure for any BS pair $(i,j)$, $i\neq j$,
$i\in\{1,\ldots,K\}$, and $j\in\{1,\ldots,K\}$, to update their
mutual IT constraints is given as follows. First, the $i$th BS
computes the elements $a$ and $b$ in $\mv{D}_{ij}$ according to
(\ref{eq:a}) and (\ref{eq:b}), respectively, with the present value
of $\mv{\Gamma}_i$. Similarly, the $j$th BS computes $c$ and $d$
with the present value of $\mv{\Gamma}_j$. Next, the $i$th BS sends
$a$ and $b$ to the $j$th BS, while the $j$th BS sends $c$ and $d$ to
the $i$th BS. Then, assuming that $\alpha_{ij}$ and $\delta_{ij}$
are preassigned values known to these two BSs, they can both compute
$\mv{d}_{ij}$ according to (\ref{eq:d update}) and update
$\Gamma_{ij}$ and $\Gamma_{ji}$ according to (\ref{eq:gamma update})
in Appendix \ref{appendix:proof conditions}. Last, with the updated
values $\Gamma'_{ij}$ and $\Gamma'_{ji}$, these two BSs reset their
respective beamforming vectors and transmit rates via solving
(\ref{eq:CR capacity}) independently. Note that the above operation
requires only local information (scalar) exchanges between different
pairs of BSs, and thus can be implemented at a very low cost in a
cellular system. One version of the decentralized algorithm for
cooperative downlink beamforming in a multi-cell system is described
in Table \ref{table}. Since in each iteration of the algorithm the
achievable rates for the pair of updating BSs both improve and those
for all other BSs are unaffected (non-decreasing), and the maximum
achievable rates for all BSs are bounded by finite Pareto-optimal
values, the convergence of this algorithm is ensured.

\begin{figure}
\centering
 \epsfxsize=0.7\linewidth
    \includegraphics[width=9cm]{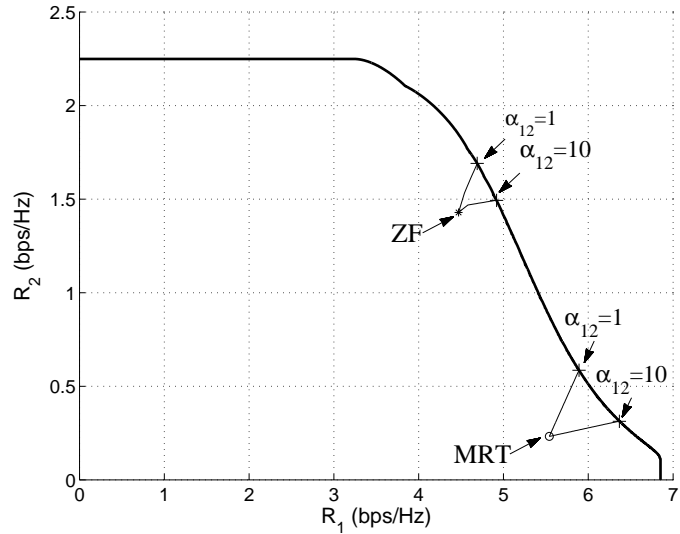}
\caption{Achievable rates for the proposed algorithm in a two-user
MISO Gaussian IC with interference treated as noise.}
\label{fig:rate region new}
\end{figure}

\begin{table}
\centering
\begin{tabular}{|l|}
\hline \hspace*{0.0cm} Initialize $\mv{\Gamma}\geq 0$ in the network
\\ \hspace*{0.0cm} BS $k$ sets $\mv{w}_k$ via solving (\ref{eq:CR
capacity}) with the given $\mv{\Gamma}_k$, $k=1,\ldots,K$ \\
\hspace*{0.0cm} Repeat \\
\hspace*{0.25cm} For $i=1,\ldots,K$, $j=1,\ldots,K, j\neq i$, \\
\hspace*{0.5cm} BS $i$ computes $a$ and $b$ in $\mv{D}_{ij}$ (cf.
(\ref{eq:a}), (\ref{eq:b})) with the given
$\mv{\Gamma}_i$ \\
\hspace*{0.5cm} Similarly, BS $j$ computes $d$ and $c$ in
$\mv{D}_{ij}$ with the given
$\mv{\Gamma}_j$ \\
\hspace*{0.5cm} BS $i$ sends $a$ and $b$ to BS $j$ \\
\hspace*{0.5cm} BS $j$ sends $c$ and $d$ to BS $i$ \\
\hspace*{0.5cm} BS $i/j$ computes $\mv{d}_{ij}$ (cf. (\ref{eq:d
update})), then updates $\Gamma_{ij}$ and $\Gamma_{ji}$ (cf.
(\ref{eq:gamma update})) \\
\hspace*{0.5cm} BS $i/j$ resets $\mv{w}_i/\mv{w}_j$ via solving
(\ref{eq:CR capacity}) with the updated
$\mv{\Gamma}_i/\mv{\Gamma}_j$ \\
\hspace*{0.25cm} End For \\
\hspace*{0.0cm} Until $|\mv{D}_{ij}|=0, \forall i\neq j$. \\
\hline
\end{tabular}
\caption{Algorithm for Cooperative Downlink Beamforming.}
\label{table}
\end{table}

\begin{example}
In Fig. \ref{fig:rate region new} (with the same two-user MISO-IC as
for Fig. \ref{fig:rate region}), we show the Pareto boundary for an
example MISO-IC with $K=2$, $M_1=M_2=3$, $P_1=5, P_2=1$, and
$\sigma_1^2=\sigma_2^2=1$, which is obtained by the proposed method
in this correspondence, i.e., solving the problems given in
(\ref{eq:CR capacity}) for $k=1,2$, and a given pair of values
$\Gamma_{12}$ and $\Gamma_{21}$ with $ 0\leq \Gamma_{12}\leq
\bar{\Gamma}_{12}$ and $ 0\leq \Gamma_{21}\leq \bar{\Gamma}_{21}$,
and then taking a closure operation over the resultant rate-pairs
with all different values of $\Gamma_{12}$ and $\Gamma_{21}$ within
their respective ranges. We demonstrate the effectiveness of the
proposed decentralized algorithm for implementing the multi-cell
cooperative downlink beamforming with two initial rate-pairs,
indicated by ``ZF'' and ``MRT'' in Fig. \ref{fig:rate region new},
which are obtained when both BSs adopt the ZF and the MRT
beamforming vectors, respectively, with their maximum transmit
powers. It is observed that the achievable rates for both MSs
increase with iterations and finally converge to a Pareto-optimal
rate-pair.\footnote{We have verified with a large number of random
channels and a variety of system parameters that the proposed
algorithm always converges to Pareto-optimal rate-pairs for the
two-user MISO-IC with randomly selected initial rate-pairs. However,
we could not prove this result in general by, e.g., showing that the
conditions given in Proposition \ref{proposition:conditions} are not
only necessary (as proved in this work) but also {\it sufficient}
for any given $\mv{\Gamma}$ to achieve a Pareto-optimal rate-tuple
for the MISO-IC.} By comparing the two cases with $\alpha_{12}=1$
and $\alpha_{12}=10$, it is observed that a larger value of
$\alpha_{12}$ results in larger rate values for the first MS in the
converged rate-pairs, which is in accordance with our previous
discussion.
\end{example}

\section{Concluding Remarks} \label{sec:conclusion}

In this correspondence, based on the concept of interference
temperature (IT) and under a cellular downlink setup, we have
developed a new method to characterize the complete Pareto boundary
of the achievable rate region for the $K$-user Gaussian MISO-IC with
interference treated as noise. It is shown that the proposed method
also leads to a new decentralized algorithm for implementing the
downlink beamforming in a cooperative multi-cell system to achieve
maximal rates with a prescribed fairness guarantee.

There are a number of directions along which the developed results
in this work can be further investigated. First, it would be
interesting to extend the multi-cell cooperative beamforming design
based on the principle of IT to the scenario where each BS supports
simultaneous transmissions to {\it multiple} active MSs each with a
single antenna or multiple antennas. Second, it remains yet to be
proved whether the necessary conditions derived in this work for any
particular set of IT constraints across the BSs to guarantee a
Pareto-optimal rate-tuple for the MISO-IC are also {\it sufficient},
even for the special two-user case. This proof is essential for the
proposed downlink beamforming algorithm to achieve the global
convergence (Pareto-optimal rates). Last but not least, it is
pertinent to analyze the proposed decentralized algorithm that
iteratively updates the mutual IT constraints between different
pairs of BSs from a {\it game-theoretical} viewpoint.

\appendices

\section{Proof of Proposition \ref{proposition:optimal CR solution}} \label{appendix:proof optimality rank one}

It can be verified that Problem (\ref{eq:CR capacity}) is convex,
and thus it can be solved by the standard Lagrange duality method
\cite{Boyd}. Let $\lambda_{kj}$, $j\neq k$, and $\lambda_{kk}$ be
the non-negative dual variables for Problem (\ref{eq:CR capacity})
associated with the $k$th BS's IT constraint for the $j$th MS and
its own transmit-power constraint, respectively. The Lagrangian
function for this problem can be written as
\begin{align}\label{eq:Lagrangian}
L(\mv{S}_k,
\mv{\lambda}_k)=\log\left(1+\frac{\mv{h}_{kk}^H\mv{S}_k\mv{h}_{kk}}{\sum_{j\neq
k}\Gamma_{jk}+\sigma_k^2}\right)-\sum_{j\neq
k}\lambda_{kj}(\mv{h}_{kj}^H\mv{S}_k\mv{h}_{kj}-\Gamma_{kj})-\lambda_{kk}({\rm
Tr}(\mv{S}_k)- P_k)
\end{align}
where $\mv{\lambda}_k=[\lambda_{k1},\ldots,\lambda_{kK}]$. The dual
function of Problem (\ref{eq:CR capacity}) is given by
\begin{align}\label{eq:dual function}
g(\mv{\lambda}_k)=\max_{\mv{S}_k\succeq 0} L(\mv{S}_k,
\mv{\lambda}_k).
\end{align}
Accordingly, the dual problem is defined as
\begin{align}\label{eq:dual problem}
\min_{\mv{\lambda_k}\geq0}~ g(\mv{\lambda}_k)
\end{align}
where $\mv{\lambda_k}\geq0$ means component-wise non-negative. Since
Problem (\ref{eq:CR capacity}) is convex with strictly feasible
points \cite{Boyd}, the duality gap between its optimal value and
that of the dual problem is zero; thus, Problem (\ref{eq:CR
capacity}) can be equivalently solved via solving its dual problem.
In order to solve the dual problem, we need to obtain the dual
function $g(\mv{\lambda}_k)$ for any given $\mv{\lambda}_k\geq 0$.
This can be done by solving the maximization problem given in
(\ref{eq:dual function}), which, according to (\ref{eq:Lagrangian}),
can be explicitly written as (by discarding irrelevant constant
terms):
\begin{align}\label{eq:dual optimization}
\mathop{\mathtt{Max.}}_{\mv{S}_k} & ~~~
\log\left(1+\frac{\mv{h}_{kk}^H\mv{S}_k\mv{h}_{kk}}{\sum_{j\neq
k}\Gamma_{jk}+\sigma_k^2}\right)-{\rm
Tr}(\mv{B}_k(\mv{\lambda}_k)\mv{S}_k)
\nonumber \\
\mathtt{s.t.} & ~~~ \mv{S}_k\succeq 0
\end{align}
where $\mv{B}_k(\mv{\lambda}_k)\triangleq \sum_{j\neq
k}\lambda_{kj}\mv{h}_{kj}\mv{h}_{kj}^H+\lambda_{kk}\mv{I}$ and
$\mv{B}_k(\mv{\lambda}_k)\succeq 0$ of dimension $M_k\times M_k$. In
order for Problem (\ref{eq:dual optimization}) to have a bounded
objective value, it is shown as follows that
$\mv{B}_k(\mv{\lambda}_k)$ should be a full-rank matrix. Suppose
that $\mv{B}_k(\mv{\lambda}_k)$ is rank-deficient, such that we
could define $\mv{S}_k=q_k\mv{v}_k\mv{v}_k^H$, where $q_k\geq 0$ and
$\mv{v}_k\in\mathbb{C}^{M_k\times 1}$ satisfying $\|\mv{v}_k\|=1$
and $\mv{B}_k(\mv{\lambda}_k)\mv{v}_k=\mv{0}$. Thereby, the
objective function of Problem (\ref{eq:dual optimization}) reduces
to
\begin{align}\label{eq:dual objective func}
\log\left(1+\frac{q_k\|\mv{h}_{kk}^H\mv{v}_k\|^2}{\sum_{j\neq
k}\Gamma_{jk}+\sigma_k^2}\right).
\end{align}
Due to the independence of $\mv{h}_{kk}$ and $\mv{h}_{kj}$'s, and
thus the independence of $\mv{h}_{kk}$ and $\mv{v}_k$, it follows
that $\|\mv{h}_{kk}^H\mv{v}_k\|>0$ with probability one such that
(\ref{eq:dual objective func}) goes to infinity by letting
$q_k\rightarrow \infty$. Since the optimal value of Problem
(\ref{eq:CR capacity}) must be bounded, without loss of generality,
we only need to consider the subset of $\mv{\lambda}_k$ in the set
$\mv{\lambda}_k\geq 0$ to make $\mv{B}_k(\mv{\lambda}_k)$ full-rank.

\begin{remark}\label{remark:KKT}
Note that from the definition of $\mv{B}_k(\mv{\lambda}_k)$ and the
Karush-Kuhn-Tucker (KKT) optimality conditions \cite{Boyd} of
Problem (\ref{eq:CR capacity}), it follows that
$\mv{B}_k(\mv{\lambda}_k)$ is full-rank only when either of the
following two cases occurs:
\begin{itemize}
\item $\lambda_{kk}>0$: In this case, the transmit power constraint
for the $k$th BS is tight for Problem (\ref{eq:CR capacity});

\item $\lambda_{kk}=0$, but there are at least $M_k$ $\lambda_{kj}$'s, $j\neq k$, having $\lambda_{kj}>0$: In this case,
regardless of whether the transmit power constraint for the $k$th BS
is tight, there are at least $M_k$ out of the $K-1$ IT constraints
of the $k$th BS are tight in Problem (\ref{eq:CR capacity}). Note
that this case can be true only when $M_k\leq K-1$.
\end{itemize}
\end{remark}

From the above discussions, it is known that
$(\mv{B}_k(\mv{\lambda}_k))^{-1}$ exists. Thus, we can introduce a
new variable $\bar{\mv{S}}_k$ for Problem (\ref{eq:dual
optimization}) as
\begin{align}\label{eq:transform}
\bar{\mv{S}}_k=(\mv{B}_k(\mv{\lambda}_k))^{1/2}\mv{S}_k(\mv{B}_k(\mv{\lambda}_k))^{1/2}
\end{align}
and substituting it into (\ref{eq:dual optimization}) yields
\begin{align}\label{eq:dual optimization new}
\mathop{\mathtt{Max.}}_{\bar{\mv{S}}_k} & ~~~
\log\left(1+\frac{\mv{h}_{kk}^H(\mv{B}_k(\mv{\lambda}_k))^{-1/2}\bar{\mv{S}}_k(\mv{B}_k(\mv{\lambda}_k))^{-1/2}\mv{h}_{kk}}{\sum_{j\neq
k}\Gamma_{jk}+\sigma_k^2}\right)-{\rm Tr}(\bar{\mv{S}}_k)
\nonumber \\
\mathtt{s.t.} & ~~~ \bar{\mv{S}}_k\succeq 0.
\end{align}
Without loss of generality, we can express $\bar{\mv{S}}_k$ into its
eigenvalue decomposition (EVD) as
$\bar{\mv{S}}_k=\mv{U}_k\mv{\Theta}_k\mv{U}_k^H$, where
$\mv{U}_k=[\mv{u}_{k1},\ldots,\mv{u}_{kM_k}]\in\mathbb{C}^{M_k\times
M_k}$ is unitary and $\mv{\Theta}_k={\rm
Diag}([\theta_{k1},\ldots,\theta_{kM_k}])\succeq 0$. Substituting
the ED of $\bar{\mv{S}}_k$ into (\ref{eq:dual optimization new})
yields
\begin{align}\label{eq:dual optimization new 2}
\mathop{\mathtt{Max.}}_{\mv{U}_k, \mv{\Theta}_k} & ~~~
\log\left(1+\frac{\sum_{i=1}^{M_k}\theta_{ki}\|\mv{h}_{kk}^H(\mv{B}_k(\mv{\lambda}_k))^{-1/2}\mv{u}_{ki}\|^2}{\sum_{j\neq
k}\Gamma_{jk}+\sigma_k^2}\right)-\sum_{i=1}^{M_k}\theta_{ki}
\nonumber \\
\mathtt{s.t.} & ~~~ \|\mv{u}_{ki}\|=1, \forall i, ~~
\mv{u}_{ki}^H\mv{u}_{kl}=0, \forall l\neq i \nonumber \\ &~~~
\theta_{ki}\geq 0, \forall i.
\end{align}
For any given $\mv{U}_k$, it can be verified that the optimal
solution of $\mv{\Theta}_k$ for Problem (\ref{eq:dual optimization
new 2}) is given by
\begin{align}\label{eq:BF vector}
\theta_{ki}=\left\{\begin{array}{ll}
\left(\frac{1}{\ln2}-\frac{\sum_{j\neq
k}\Gamma_{jk}+\sigma_k^2}{\|\mv{h}_{kk}^H(\mv{B}_k(\mv{\lambda}_k))^{-1/2}\mv{u}_{ki}\|^2}\right)^+
& {\rm if} ~ i=\mathop{\arg}\max_{l\in\{1,\ldots,M_k\}}
\|\mv{h}_{kk}^H(\mv{B}_k(\mv{\lambda}_k))^{-1/2}\mv{u}_{kl}\| \\ 0 &
{\rm otherwise}. \end{array} \right.
\end{align}
Thus, it follows that for the optimal solution of Problem
(\ref{eq:dual optimization new}), ${\rm Rank}(\bar{\mv{S}}_k)\leq
1$. Furthermore, let $i'$ denote the index of $i$ for which
$\theta_{ki'}\geq 0$. The objective function of Problem
(\ref{eq:dual optimization new 2}) reduces to
\begin{align}
\log\left(1+\frac{\theta_{ki'}\|\mv{h}_{kk}^H(\mv{B}_k(\mv{\lambda}_k))^{-1/2}\mv{u}_{ki'}\|^2}{\sum_{j\neq
k}\Gamma_{jk}+\sigma_k^2}\right)-\theta_{ki'}.
\end{align}
Clearly, the above function is maximized with any $\theta_{ki'}>0$
when
\begin{align}\label{eq:power}
\mv{u}_{ki'}=\frac{(\mv{B}_k(\mv{\lambda}_k))^{-1/2}\mv{h}_{kk}}{\|(\mv{B}_k(\mv{\lambda}_k))^{-1/2}\mv{h}_{kk}\|}.
\end{align}
From (\ref{eq:BF vector}) and (\ref{eq:power}), it follows that the
optimal solution for Problem (\ref{eq:dual optimization new}) is
\begin{align}
\bar{\mv{S}}_k=\frac{\left(\frac{1}{\ln2}-\frac{\sum_{j\neq
k}\Gamma_{jk}+\sigma_k^2}{\|\mv{h}_{kk}^H(\mv{B}_k(\mv{\lambda}_k))^{-1/2}\|^2}\right)^+}
{\|(\mv{B}_k(\mv{\lambda}_k))^{-1/2}\mv{h}_{kk}\|^2}(\mv{B}_k(\mv{\lambda}_k))^{-1/2}\mv{h}_{kk}\mv{h}_{kk}^H(\mv{B}_k(\mv{\lambda}_k))^{-1/2}.
\end{align}
Combining the above solution and (\ref{eq:transform}), it can be
easily shown that the optimal solution $\mv{S}_k$ for Problem
(\ref{eq:CR capacity}) is as given by Proposition
\ref{proposition:optimal CR solution}.

With the obtained dual function $g(\mv{\lambda}_k)$ for any given
$\mv{\lambda_k}$, the dual problem (\ref{eq:dual problem}) can be
solved by searching over $\mv{\lambda_k}\geq0$ to minimize
$g(\mv{\lambda}_k)$. This can be done via, e.g., the ellipsoid
method \cite{BGT81}, by utilizing the subgradient of
$g(\mv{\lambda}_k)$ that is obtained from (\ref{eq:Lagrangian}) as
$\Gamma_{kj}-\mv{h}_{kj}^H\mv{S}^*_k\mv{h}_{kj}$ for $\lambda_{kj},
j\neq k$ and $P_k-{\rm Tr}(\mv{S}_k^*)$ for $\lambda_{kk}$, where
$\mv{S}^*_k$ is the optimal solution for Problem (\ref{eq:dual
optimization}) with the given $\mv{\lambda}_k$. When
$\mv{\lambda_k}$ converges to the optimal solution for the dual
problem, the corresponding $\mv{S}^*_k$ becomes the optimal solution
for Problem (\ref{eq:CR capacity}). Proposition
\ref{proposition:optimal CR solution} thus follows.

\section{Proof of Proposition \ref{proposition:relationship}} \label{appendix:proof relationship}

Since the given set of $\mv{S}_1,\ldots,\mv{S}_K$ achieves the
Pareto-optimal rate-tuple  $(R_1,\ldots,R_K)$ for the MISO-IC, from
(\ref{eq:user rates}) and (\ref{eq:rate region}) it follows that for
any $k\in\{1,\ldots,K\}$
\begin{align}\label{eq:Rk}
R_k=\log\left(1+\frac{\mv{h}_{kk}^H\mv{S}_k\mv{h}_{kk}}{\sum_{j\neq
k}\mv{h}_{jk}^H\mv{S}_j\mv{h}_{jk}+\sigma_k^2}\right).
\end{align}
Since $\Gamma_{jk}=\mv{h}_{jk}^H\mv{S}_j\mv{h}_{jk}, \forall j\neq
k$, (\ref{eq:Rk}) can be rewritten as
\begin{align}\label{eq:Rk new}
R_k=\log\left(1+\frac{\mv{h}_{kk}^H\mv{S}_k\mv{h}_{kk}}{\sum_{j\neq
k}\Gamma_{jk}+\sigma_k^2}\right).
\end{align}
Note that (\ref{eq:Rk new}) is the same as the objective function of
Problem (\ref{eq:CR capacity}) for the given $k$. Furthermore, from
(\ref{eq:rate region}) it follows that ${\rm Tr}(\mv{S}_k)\leq P_k$.
Using this and the fact that
$\Gamma_{kj}=\mv{h}_{kj}^H\mv{S}_k\mv{h}_{kj}, \forall j\neq k$, it
follows that $\mv{S}_k$  satisfies the constraints given in Problem
(\ref{eq:CR capacity}) for the given $k$. Therefore, $\mv{S}_k$ must
be a feasible solution for Problem (\ref{eq:CR capacity}) with the
given $k$ and $\mv{\Gamma}_k$.

Next, we need to prove that $\mv{S}_k$ is indeed the optimal
solution of Problem (\ref{eq:CR capacity}) for any given $k$, and
thus the corresponding achievable rate $R_k$ is equal to the optimal
value of Problem (\ref{eq:CR capacity}), which is
$C_k(\mv{\Gamma}_k)$. We prove this result by contradiction. Suppose
that the optimal solution for Problem (\ref{eq:CR capacity}),
denoted by $\mv{S}_k^{\star}$, is not equal to $\mv{S}_k$ for a
given $k$. Thus, we have
\begin{align}
R_k<&\log\left(1+\frac{\mv{h}_{kk}^H\mv{S}_k^{\star}\mv{h}_{kk}}{\sum_{j\neq
k}\Gamma_{jk}+\sigma_k^2}\right) \\
=&
\log\left(1+\frac{\mv{h}_{kk}^H\mv{S}_k^{\star}\mv{h}_{kk}}{\sum_{j\neq
k}\mv{h}_{jk}^H\mv{S}_j\mv{h}_{jk}+\sigma_k^2}\right)\triangleq r_k.
\end{align}

Furthermore, since $\mv{h}_{kj}^H\mv{S}_k^{\star}\mv{h}_{kj}\leq
\Gamma_{kj}, \forall j\neq k$, we have for any $j\neq k$,
\begin{align}\label{eq:Rj}
R_j=&\log\left(1+\frac{\mv{h}_{jj}^H\mv{S}_j\mv{h}_{jj}}{\sum_{i\neq
j}\mv{h}_{ij}^H\mv{S}_i\mv{h}_{ij}+\sigma_j^2}\right) \\
=&\log\left(1+\frac{\mv{h}_{jj}^H\mv{S}_j\mv{h}_{jj}}{\sum_{i\neq
j}\Gamma_{ij}+\sigma_j^2}\right) \\
\leq &
\log\left(1+\frac{\mv{h}_{jj}^H\mv{S}_j\mv{h}_{jj}}{\sum_{i\neq j,
k}\Gamma_{ij}+\mv{h}_{kj}^H\mv{S}_k^{\star}\mv{h}_{kj}+\sigma_j^2}\right)
\\
=& \log\left(1+\frac{\mv{h}_{jj}^H\mv{S}_j\mv{h}_{jj}}{\sum_{i\neq
j,k}\mv{h}_{ij}^H\mv{S}_i\mv{h}_{ij}+\mv{h}_{kj}^H\mv{S}_k^{\star}\mv{h}_{kj}+\sigma_j^2}\right)
\triangleq r_j.
\end{align}
Thus, for another set of transmit covariance matrices given by
$\mv{S}_1,\ldots,\mv{S}_{k-1},\mv{S}_k^{\star},\mv{S}_{k+1},\ldots,\mv{S}_K$,
the corresponding achievable rate-tuple for the MISO-IC,
$(r_1,\ldots,r_K)$, satisfies that $r_k>R_k$ and $r_j\geq R_j,
\forall j\neq k$, which contradicts the fact that $(R_1,\ldots,R_K)$
is a Pareto-optimal rate-tuple for the MISO-IC. Hence, the
presumption that $\mv{S}_k\neq\mv{S}_k^{\star}$ for any given $k$
cannot be true. Thus, we have $\mv{S}_k=\mv{S}_k^{\star}$ and
$R_k=C_k(\mv{\Gamma}_k), \forall k$. Proposition
\ref{proposition:relationship} thus follows.

\section{Proof of Proposition \ref{proposition:conditions}} \label{appendix:proof conditions}

As given in Proposition \ref{proposition:conditions}, with
$\mv{\Gamma}$, the corresponding optimal values of the problems in
(\ref{eq:CR capacity}) for all $k$'s, $C_k(\mv{\Gamma}_k)$'s,
correspond to a Pareto-optimal rate-tuple for the MISO-IC, denoted
by $(R_1,\ldots,R_K)$. Let $\mv{S}_1,\ldots,\mv{S}_K$ denote the set
of optimal solutions for the problems in (\ref{eq:CR capacity}). We
thus have
\begin{align}\label{eq:R k renew}
C_k(\mv{\Gamma}_k)=R_k=\log\left(1+\frac{\mv{h}_{kk}^H\mv{S}_k\mv{h}_{kk}}{\sum_{j\neq
k}\Gamma_{jk}+\sigma_k^2}\right), ~ k=1,\ldots,K.
\end{align}
Next, we prove Proposition \ref{proposition:conditions} by
contradiction. Suppose that there exists a pair of $(i,j)$ with
$|\mv{D}_{ij}|\neq 0$, where $\mv{D}_{ij}$ is defined in
(\ref{eq:D}). Define a new $\mv{\Gamma}'$ over $\mv{\Gamma}$, where
all the elements in $\mv{\Gamma}$ remain unchanged except
$[\Gamma_{ij}, \Gamma_{ji}]^T$ being replaced by
\begin{align}\label{eq:gamma update}
[\Gamma'_{ij}, \Gamma'_{ji}]^T=[\Gamma_{ij},
\Gamma_{ji}]^T+\delta_{ij}\cdot\mv{d}_{ij}
\end{align}
where $\delta_{ij}>0$ is a small step-size, and $\mv{d}_{ij}$ is any
vector that satisfies $\mv{D}_{ij}\mv{d}_{ij}> 0$ (component-wise),
with one possible value for such $\mv{d}_{ij}$ is given by
(\ref{eq:d update}) in the main text. With $\mv{\Gamma}'$, the
optimal solutions for the problems in (\ref{eq:CR capacity}) remain
unchanged $\forall k\neq i,j$, while for those with $k=i$ and $k=j$,
the optimal solutions are changed to be $\mv{S}_i^{\star}$ and
$\mv{S}^{\star}_j$, respectively. Accordingly, the new achievable
rates in the MISO-IC for any $k\neq i,j$ are given by
\begin{align}
r_k=&\log\left(1+\frac{\mv{h}_{kk}^H\mv{S}_k\mv{h}_{kk}}{\sum_{l\neq
k,i,j}\mv{h}_{lk}^H\mv{S}_l\mv{h}_{lk}+\mv{h}_{ik}^H\mv{S}_i^{\star}\mv{h}_{ik}+\mv{h}_{jk}^H\mv{S}_j^{\star}\mv{h}_{jk}
+\sigma_k^2}\right) \\
=&\log\left(1+\frac{\mv{h}_{kk}^H\mv{S}_k\mv{h}_{kk}}{\sum_{l\neq
k,i,j}\Gamma_{lk}+\mv{h}_{ik}^H\mv{S}_i^{\star}\mv{h}_{ik}+\mv{h}_{jk}^H\mv{S}_j^{\star}\mv{h}_{jk}
+\sigma_k^2}\right) \\
\geq & R_k \label{eq:inequality}
\end{align}
where (\ref{eq:inequality}) is due to  (\ref{eq:R k renew}) and the
facts that $\mv{h}_{ik}^H\mv{S}_i^{\star}\mv{h}_{ik}\leq
\Gamma_{ik}$ and $\mv{h}_{jk}^H\mv{S}_j^{\star}\mv{h}_{jk}\leq
\Gamma_{jk}$. Also, it can be shown that
\begin{align}
r_i=&\log\left(1+\frac{\mv{h}_{ii}^H\mv{S}_i^{\star}\mv{h}_{ii}}{\sum_{l\neq
i,j}\mv{h}_{li}^H\mv{S}_l\mv{h}_{li}+\mv{h}_{ji}^H\mv{S}_j^{\star}\mv{h}_{ji}
+\sigma_i^2}\right) \\
=&\log\left(1+\frac{\mv{h}_{ii}^H\mv{S}_i^{\star}\mv{h}_{ii}}{\sum_{l\neq
i,j}\Gamma_{li}+\mv{h}_{ji}^H\mv{S}_j^{\star}\mv{h}_{ji}+\sigma_i^2}\right) \\
\geq & C_i(\mv{\Gamma}'_i) \label{eq:inequality 2}
\end{align}
where (\ref{eq:inequality 2}) is due to the facts that
$\mv{h}_{ji}^H\mv{S}_j^{\star}\mv{h}_{ji}\leq \Gamma'_{ji}$ and
$\mv{S}_i^{\star}$ achieves the optimal value of Problem (\ref{eq:CR
capacity}) with $k=i$ and the given $\mv{\Gamma}'_i$, denoted by
$C_i(\mv{\Gamma}'_i)$. Similarly, it can be shown that $r_j\geq
C_j(\mv{\Gamma}'_j)$. Thus, from (\ref{eq:gamma update}) and
$\mv{D}_{ij}\mv{d}_{ij}> 0$, it follows that with sufficiently small
$\delta_{ij}$
\begin{align}
\left[\begin{array}{l} r_i \\ r_j \end{array} \right] \geq&
\left[\begin{array}{l} C_i(\mv{\Gamma}'_i)  \\
C_j(\mv{\Gamma}'_j)\end{array} \right]\\
\cong & \left[\begin{array}{l} C_i(\mv{\Gamma}_i)  \\
C_j(\mv{\Gamma}_j)\end{array}
\right]+\delta_{ij}\mv{D}_{ij}\mv{d}_{ij} \\
> & \left[\begin{array}{l} R_i \\ R_j \end{array} \right].
\end{align}
Therefore, we have found a new set of achievable rate-tuple for the
MISO-IC with $\mv{\Gamma}'$, $(r_1,\ldots,r_K)$, which has $r_i>
R_i$, $r_j> R_j$, and $r_k\geq R_k, \forall k\neq i,j$. Clearly,
this contradicts the fact that $(R_1,\ldots,R_K)$ is Pareto-optimal
for the MISO-IC. Thus, the presumption that there exists a pair of
$(i,j)$ with $|\mv{D}_{ij}|\neq 0$ cannot be true. Proposition
\ref{proposition:conditions} thus follows.

\end{document}